\documentclass[conference, 10pt]{IEEEtran}
\usepackage{acronym}

\usepackage{amsthm}

\newtheorem{lemma}{Lemma}
\newtheorem{proposition}{Proposition}

\newtheorem{definition}{Definition}

\newtheorem{Rem}{Remark}

\usepackage{enumitem}
\newcounter{problem}
\newcounter{subproblem}[problem]

\usepackage{hyperref}
\usepackage{cleveref}

\usepackage{mathtools}

\usepackage{color}
\usepackage{xcolor}
\usepackage{colortbl}
\definecolor{purple}{RGB}{139, 0, 139}
\usepackage{manfnt}
\usepackage{cleveref}
\usepackage{url}

\newif\iftodo   
\todotrue
\newif\iftodoshort  
\todoshortfalse
\ifCLASSINFOpdf
  \usepackage[pdftex]{graphicx}
   \DeclareGraphicsExtensions{.pdf,.jpeg,.png,.tiff}
\else
   \usepackage[dvips]{graphicx}
	 \usepackage{epsfig}
   \usepackage{import}
   \DeclareGraphicsExtensions{.eps,.ps}
\fi
\usepackage{amssymb,amsmath}
\usepackage[mathscr]{euscript}
\usepackage{bbm}
\usepackage{dsfont}

\DeclareMathOperator*\dom{dom}		
\DeclareMathOperator*\maximize{maximize}		
\DeclareMathOperator*\subject{subject \ to}		

\usepackage{algorithmic}
\usepackage[ruled,vlined,linesnumbered]{algorithm2e}

\usepackage{cite}
 \usepackage{caption}
 \usepackage{subcaption}
\graphicspath{ {./fig/} }
\newcommand{\Rmnum}[1]{\uppercase\expandafter{\romannumeral #1}}
\newcommand{\rmnum}[1]{\lowercase\expandafter{\romannumeral #1}}

\newcommand{\diag}{\mathop{\mathrm{diag}}}

\newcommand{\field}[1]{\mathbb{#1}}

\newcommand{\emenge}[1]{\mathscr{#1}}
\newcommand{\set}[1]{\mathscr{#1}}

\newcommand{\operator}[1]{\mathrm{#1}}

\newcommand{\R}{{\field{R}}}
 
\newcommand{\RN}{{\field{R}}_{+}}
\newcommand{\RP}{{\field{R}}_{++}}

\newcommand{\Ns}{{\emenge{N}}}
\newcommand{\Ks}{{\emenge{K}}}

\newcommand{\sinr}{\operator{SINR}}

\newcommand{\ma}[1]{\boldsymbol{\mathbf{#1}}}
\newcommand{\ve}[1]{\boldsymbol{\mathbf{#1}}}

\newcommand{\vx}{\ve{x}}

\newcommand{\vw}{\ve{w}}

\newcommand{\vp}{\ve{p}}

\newcommand{\vy}{\ve{y}}

\newcommand{\vf}{\ve{f}}

\newcommand{\mX}{\ma{X}}

\newcommand{\mG}{\ma{G}}

\newcommand{\mA}{\ma{A}}

\newcommand{\mV}{\ma{V}}
\newcommand{\mVt}{\tilde{\ma{V}}}

\newcommand{\mD}{\ma{D}}
\newcommand{\mC}{\ma{C}}

\usepackage[utf8]{inputenc}

\newcommand{\ul}{(\text{u})}
\newcommand{\dl}{(\text{d})}

\newcommand{\tmax}{\text{max}}
\newcommand{\trans}{\text{(trans)}}

\usepackage{hyperref}
\usepackage{amsthm}
\usepackage{amssymb,amsmath}
\usepackage[mathscr]{euscript}

\newcommand{\cosl}[1]{}
\newcommand{\resl}[1]{}

\usepackage[draft,footnote,nomargin]{fixme}
\newcommand{\fnql}[1]{}
\newcommand{\fnsv}[1]{}

\begin{document}
%
\title{Improving Resource Efficiency with Partial Resource Muting for Future Wireless Networks}
\author{
\IEEEauthorblockN{Qi Liao\IEEEauthorrefmark{1} and R. L. G. Cavalcante\IEEEauthorrefmark{2}}
\IEEEauthorblockA{\IEEEauthorrefmark{1}Nokia Bell Labs, Stuttgart, Germany\\ 
\IEEEauthorrefmark{2}Fraunhofer Heinrich Hertz Institute and Technical University of Berlin, Germany \\
Email: qi.liao@nokia-bell-labs.com; renato.cavalcante@hhi.fraunhofer.de}
}
\maketitle
\begin{abstract} 
We propose novel resource allocation algorithms that have the objective of finding a good tradeoff between resource reuse and interference avoidance in wireless networks. To this end, we first study properties of functions that relate the resource budget available to network elements to the optimal utility and to the optimal resource efficiency obtained by solving max-min utility optimization problems. From the asymptotic behavior of these functions, we obtain a transition point that indicates whether a network is  operating in an efficient noise-limited regime or in an inefficient interference-limited regime for a given resource budget. For networks operating in the inefficient regime, we propose a novel partial resource muting scheme to improve the efficiency of the resource utilization. The framework is very general. It can be applied not only to the downlink of 4G networks, but also to  5G networks equipped with flexible duplex mechanisms. Numerical results show significant performance gains of the proposed scheme compared to the solution to the max-min utility optimization problem with full frequency reuse.
\end{abstract}
\section{Introduction}\label{sec:intro}
Future network architectures are envisioned to provide seamless connections anywhere and anytime with asymmetric traffic. In particular, in the \ac{MAC}/link layer, one of the key challenges to bring this vision to reality is to devise service-centric resource allocation mechanisms able to find a good balance between interference avoidance and resource (spectrum) reuse. Promising approaches to address this challenge are the muting schemes, which include the \ac{ABS} approach in time domain \cite{3GPP36133} and the mutually exclusive resource block allocation approach in frequency domain \cite{hegde2016optimal}. However, in these approaches two questions remain largely unanswered:
\begin{itemize}
\item At which point full resource reuse becomes inefficient? In other words, when is it appropriate to activate resource muting?  
\item How to develop efficient resource muting schemes that deal with complex interference patterns caused by disruptive architectures such as flexible duplex in future networks? 
\end{itemize} 
In this paper, we address these questions by studying properties of solutions to a large class of max-min utility fairness problems. In more detail, building upon the seminal work of Yates \cite{yates1995framework}, which has introduced the concept of {\it \acp{SIF}} in wireless networks (see Definition \ref{def:SIF}), many researchers have devised efficient algorithms able to solve a large class of max-min fairness problems with the \ac{SINR} \cite{yates1995integrated,vucic2011fixed,zheng2016wireless,nuzman2007contraction} as the utility function. More recently, fairness problems with rate-related utilities characterized by nonlinear load coupling models have also been increasingly gaining attention \cite{ho2015power,cavalcante2017max,zheng2016wireless}. Most studies devoted to these problems focus heavily on the development of numerical solvers, and, by doing so, issues such as properties of the utility and the network efficiency as a function of the budget available to network elements have remained largely unexplored.  This gap has been addressed in the recent study in \cite{cavalcante2017performance}, which has used the concept of asymptotic functions in nonlinear analysis \cite{auslender2006asymptotic} to derive tight bounds for the network performance obtained by solving max-min utility problems. That study has also derived a transition point that, for a given resource budget, indicates whether a network operates in a resource efficient region.

\par In this study we build upon the findings in \cite{cavalcante2017performance} to develop novel efficient resource allocation algorithms for current and future network technologies.  Our main contributions can be summarized as follows:
\begin{itemize}
\item  We develop partial resource muting schemes based on the characterization of  resource-efficient regions of solutions to a general class of resource allocation problems.
\item The above schemes, which are based on the solutions to a series of subproblems using fixed point iterations, efficiently detect a set of bottleneck users that should be assigned to the muting region, and the schemes optimize the service-centric resource allocation. 
\item The framework proposed here can be applied both to the \ac{DL} of 4G  networks and to 5G networks equipped with flexible duplex mechanisms. 
\end{itemize}
This study is organized as follows. In Section \ref{sec:def} we introduce notation and mathematical background on \rmnum{1})  max-min fairness problems and {\it \acp{CEVP}}, and \rmnum{2}) properties of the solutions to \acp{CEVP}. In Section \ref{sec:DL} we propose a resource muting scheme for \ac{DL} resource allocation. This scheme is enhanced in Section \ref{sec:Extension_5G} to cater for complex interference models in 5G networks. Conclusions are summarized in Section \ref{sec:concl}.
\section{Notations, Definitions, and Preliminaries}\label{sec:def}
The following definitions are used in this paper. The nonnegative and positive orthant in $K$ dimensions are denoted by $\RN^K$ and $\RP^K$, respectively. Let $\vx\leq\vy$ denote the component-wise inequality between two vectors. A norm $\|\cdot\|$ on $\R_+^K$ is said to be monotone if $(\forall \vx\in\R_+^K)(\forall \vy\in \R_+^K) \ \ve{0}\leq \vx\leq \vy \Rightarrow \|\vx\|\leq \|\vy\|$. By $\diag(\vx)$ we denote a diagonal matrix with the elements of $\vx$ on the main diagonal.
The cardinality of set $\set{A}$ is denoted by $|\set{A}|$. The positive part of a real function is defined by $\left[f(x)\right]^+:=\max\left\{0, f(x)\right\}$. Standard interference functions (SIFs) are defined as follows:

\begin{definition}{\cite{yates1995framework}}
	\label{def:SIF}
	A function $f : \R^K\to \RP\cup\{\infty\}$ is an \ac{SIF} if the following axioms hold:
	1) (Monotonicity) $\left(\forall \vx\in\RN^K\right)\left(\forall \vy\in\RN^K\right)$ \ $\vx \leq \vy \Rightarrow f(\vx)\leq f(\vy)$; 
	2) (Scalability) $\left(\forall \vx\in\RN^K\right)\left(\forall \alpha>1\right)$ \ $\alpha f(\vx)> f(\alpha\vx)$;  
	and 3) (Nonnegative effective domain) $\dom f := \{\vx\in\mathbb{R}^K~|~f(\vx)<\infty\}=\RN^K$.
	A vector function $\vf:\RN^K\to\RP^K:\vx\mapsto[f_1(\vx), \ldots, f_N(\vx)]$ is called an \ac{SIF} if each of the component functions is an \ac{SIF}.
\end{definition}


We now turn our attention to the general description of the problems we address in this study. In more detail, a large array of utility maximization problems in wireless networks can be seen as particular instances of the following optimization problem \cite{yates1995integrated,vucic2011fixed,zheng2016wireless,ho2015power}: 
\vspace{-1.5ex}
\begin{subequations}\label{eqn:Maxmin}
\begin{align}
\maximize_{\vx\in\RN^K} \ & \min_{k\in\Ks} \  u_k(\vx) \label{eqn:Maxmin_1}\\
\subject \ & \|\vx\|\leq \theta \label{eqn:Maxmin_2}
\vspace{-1ex}
\end{align}
\end{subequations}
where $\Ks:=\{1, \ldots, K\}$ is the set of network elements, $u_k:\RN^K\to\RN$ is the utility function of network element $k\in\Ks$,  and $\|\cdot\|$ is a monotone norm used to constrain the resource utilization $\vx=[x_1,\ldots,x_K]$ to a given budget $\theta> 0$. In the above problem formulation, the $k$th component $x_k$ of the optimization variable $\vx$ is the resource utilization of network element $k\in\Ks$.  Now, consider the following conditional eigenvalue problem (CEVP):
%

{({\bf The conditional eigenvalue problem}:) \it Given a monotone norm $\|\cdot\|$, a budget $\theta\in\RP$, and a mapping $\ve{T}:\RN^K\to\RP^K:\vx\mapsto\left[T^{(1)}(\vx), \ldots, T^{(K)}(\vx)\right]$, where $T^{(k)}:\RN^K\to\RP$ is an SIF for each $k\in\Ks$, the CEVP is stated as follows:  
	\begin{align}
	\label{eqn:cevp}
	\text{Find } & (\vx,c)\in\RN^K\times\RP\nonumber\\
	\text{such that }& \ve{T}(\vx) = \dfrac{1}{c} \vx \text{ and }\|\vx\|=\theta.
	\end{align}}
As an implication of the results in \cite{nuzman2007contraction}, if the utility functions in \eqref{eqn:Maxmin} and the SIF $\ve{T}$ in \eqref{eqn:cevp} are related by $(\forall k\in\Ks)(\forall \vx\in\R^K)~u_k(\vx)=x_k/T^{(k)}(\vx)$, then $\vx^\star$ solves \eqref{eqn:Maxmin} if $(\vx^\star,c^\star)$ solves \eqref{eqn:cevp}. Furthermore, we have $(\forall k\in \Ks)~u_k(\vx^\star)=c^\star$. Therefore, to solve \eqref{eqn:Maxmin}, we only need to devise efficient algorithmic solutions to \eqref{eqn:cevp}. To this end, \cite{nuzman2007contraction} has proved that \eqref{eqn:cevp} has a unique solution $(\vx^\star,c^\star)\in\RP^K\times\RP$ and that $\vx^\star\in\RP^K$ is the limit of the sequence $\left(\vx^{(n)}\right)_{n\in\mathbb{N}}$ generated by
\begin{equation}
\vx^{(n+1)} = \frac{\theta}{\left\|\ve{T}\left(\vx^{(n)}\right)\right\|}\ve{T}\left(\vx^{(n)}\right), \mbox{ with } \vx^{(0)}\in\RN^K.
\label{eqn:FP_w}
\vspace{-.5ex}
\end{equation}
With knowledge of $\vx^\star$, we recover $c^\star$ from the equality $c^\star=\theta/\|\ve{T}(\vx^\star)\|$.

Note that later studies have also established the connection between Problem \eqref{eqn:Maxmin} and Problem \eqref{eqn:cevp} by using arguments with different levels of generality \cite{zheng2016wireless}. Recently, \cite{cavalcante2017performance} has studied the influence of the budget $\theta$ on the solution to these problems, and we summarize some of the results of that study because they are crucial to the contributions that follow. 

One of the key tools used in the analysis in \cite{cavalcante2017performance} is the notion of asymptotic functions associated with \acp{SIF}: 
%
\vspace{-1ex}
\begin{proposition}{\cite[Prop. 1]{cavalcante2017performance}}
In $\RN^{K}$, the asymptotic function $T_{\infty}: \R^K\to \R\cup\{\infty\}$ associated with an \ac{SIF} $T:\R^K\to \RP\cup\{\infty\}$ is given by
	\vspace{-1ex}
  \begin{equation}
	\left(\forall\vx\in \RN^K\right) T_{\infty} (\vx) = \lim_{h\to \infty} T(h\vx)/h \in\RN.
	\label{eqn:asympSIF}
	\end{equation}
\label{prop:asymp} 
\vspace{-4ex}
\end{proposition}
\vspace{-.5ex}
 
For convenience, before we show relations between properties of asymptotic functions and properties of the solution to Problem \eqref{eqn:cevp}, let us the recall the following definitions:


%

%
\vspace{-.5ex}
\begin{definition}[Utility and resource efficiency]\cite[Def. 4]{cavalcante2017performance}
Let $(\vx_{\theta}, c_{\theta})\in\RP^K\times\RP$ denote the solution to Problem \eqref{eqn:cevp} for a given budget $\theta\in\RP$. The utility and the resource efficiency function are defined by, respectively, $U:\RP\to \RP:\theta\mapsto c_{\theta}$ and $E:\RP\to\RP:\theta\mapsto U(\theta)/\|\vx_{\theta}\|$.
\label{def:UtilityEfficiency}
\end{definition}
We can now state selected properties of the solution to Problem~\eqref{eqn:cevp} (and hence to Problem~\eqref{eqn:Maxmin}):
\begin{proposition}{\cite[Prop. 3]{cavalcante2017performance}}
Assume that the following \ac{CEVP}:
Find  $(\vx_{\infty}, \lambda_{\infty})\in\RN^K\times\RN$ such that 
	\vspace{-.5ex}
\begin{equation}\label{eqn:AsymFPI}
\ve{T}_{\infty}(\vx_{\infty})  = \lambda_{\infty}\vx_{\infty}, \
\|\vx_{\infty}\|  = 1
	\vspace{-.5ex}
\end{equation}
has a unique positive solution $(\vx_{\infty}, \lambda_{\infty})\in\RP^K\times\RP$, where
each coordinate $T^{(k)}_{\infty}$  of $\ve{T}_\infty:\RN^K\to\RN^K$ is an asymptotic function associated with the \ac{SIF} $T^{(k)}$ in \eqref{eqn:cevp}. 
Then the solution to \eqref{eqn:cevp} has the following properties:
\begin{itemize}
\item[{\normalfont (i)}] Asymptotic utility and resource efficiency:
\begin{align}
\vspace{-.5ex}
 \sup_{\theta>0} U(\theta) & = \lim_{\theta\to\infty} U(\theta) = 1/\lambda_{\infty},\label{eqn:asympU}\\
 \sup_{\theta>0} E(\theta) & = \lim_{\theta\to 0^{+}} E(\theta) = 1/\|\ve{T}(\ve{0})\|.\label{eqn:asympE}
\vspace{-.5ex}
\end{align}
\item[{\normalfont (ii)}] Upper bound for the utility: $(\forall \theta\in\RP)$,
  \begin{equation}
 U(\theta)\leq 
\begin{cases}
\vspace{-.5ex}
\theta/\|\ve{T}(\ve{0})\|, \mbox{  if } \theta\leq \theta^{\trans}\\
1/\lambda_{\infty}, \mbox{ otherwise},
\vspace{-.5ex}
\end{cases}	
\label{eqn:utilityUB}
	\end{equation}
	where $\theta^{\trans}$ is the \emph{transition point} defined by 
	\begin{equation}
	\vspace{-.5ex}
	\theta^{\trans} := \|\ve{T}(\ve{0})\|/\lambda_{\infty}.
	\label{eqn:TransPoint}
	\vspace{-.5ex}
	\end{equation}
\item[{\normalfont (iii)}] Upper bound for the resource efficiency: $(\forall \theta\in\RP)$,
\begin{equation}
E(\theta)\leq \min\left\{1/\|\ve{T}(\ve{0})\|, 1/(\lambda_{\infty}\theta) \right\}.
\label{eqn:energyeffUB}
\end{equation}
\end{itemize}
\label{prop:PerformLimits}
\end{proposition} 
The transition point $\theta^{\trans}$ in \eqref{eqn:TransPoint} serves as a coarse indicator of whether we can obtain substantial gains in utility by increasing the budget $\theta$. More precisely, if the given budget $\theta$ is greater than $\theta^{\trans}$, then the network is likely operating in a regime where the performance is limited by interference, so increasing $\theta$ even by orders of magnitude typically brings only marginal gains in utility. In contrast, if $\theta < \theta^{\trans}$, then noticeable gains in utility can be obtained by increasing the budget $\theta$. Fig. \ref{fig:PowerUtility_Bound} illustrates this observation, which is heavily exploited by the algorithms proposed in the next sections.

%
 \begin{figure}[t]
    \centering
		    \includegraphics[width=1\columnwidth]{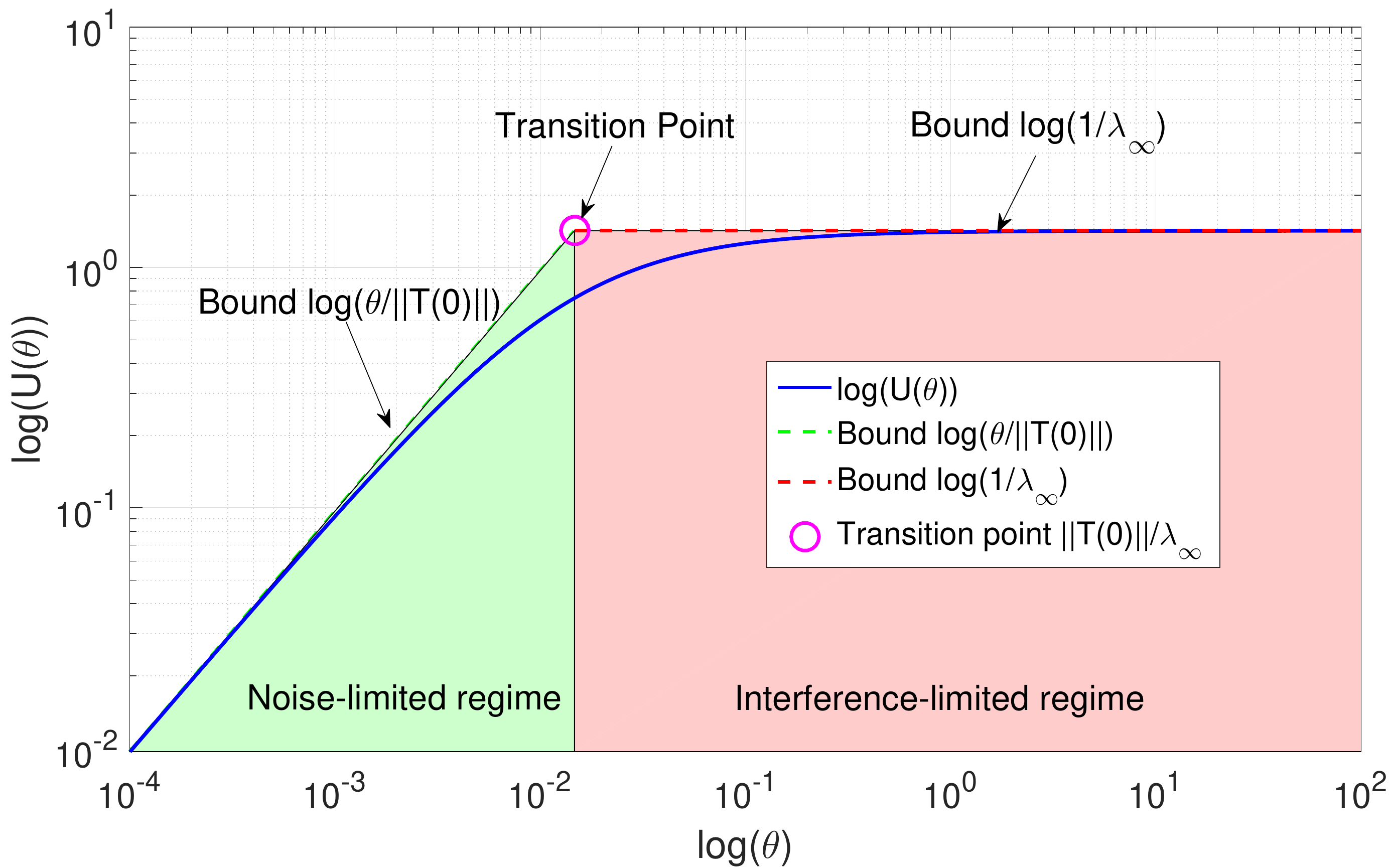}
				\vspace{-1ex}
        \caption{Example: utility as a function of the power budget $\theta$.}
				\vspace{-3ex}
 \label{fig:PowerUtility_Bound}
\end{figure}
\vspace{-1ex}
\section{Proposed algorithms for DL resource allocation}
\label{sec:DL}
%
In this section, we build upon the results in \cite{nuzman2007contraction} and Proposition \ref{prop:PerformLimits} to devise novel algorithms for DL resource allocation. Then, in Section~\ref{sec:Extension_5G}, we study problems with more complex interference models incorporating disruptive architectures envisioned for 5G networks.
\vspace{-.5ex}
\subsection{System Model and Problem Formulation}\label{subsec:DL_SysModel}
We consider an \ac{OFDM}-based wireless network system, consisting of a set of \acp{BS} $\Ns:=\{1, \ldots, N\}$ and a set of communication links $\Ks:=\{1, \ldots, K\}$ in \ac{DL}. Hereafter, we use the terms \lq\lq communication links\rq\rq, \lq\lq services\rq\rq,  and \lq\lq users\rq\rq \ interchangeably (without loss of generality, we can assume that each user sets up one communication link for one service at an unit of time). Let $\overline{W}$ denote the total bandwidth in Hz, and let $\vw\in[0,1]^K$ be a vector collecting the fraction of bandwidth allocated to the services.  We assume that the transmit power spectral density (in Watts/Hz) of all services are given and collected in a vector $\vp\in\RP^K$. Let the data rate demand of all services (in bit/s) be collected in $\bar{\ve{r}}\in\RP^K$.
The matrix $\mA\in\{0,1\}^{N\times K}$ denotes the \ac{BS}-to-service assignment matrix, where $a_{n,k}=1$ if $k$ is served by \ac{BS} $n$, and $0$ otherwise.


By $v_{k,l}$ we denote the channel gain between the transmitter of service $l$ and the receiver of service $k$. For $k\neq l$, $v_{k,l}$ is the interference channel gain, which is positive if service $l$ causes interference to service $k$, otherwise $v_{k,l} = 0$. The scalar $v_{k,k}>0$ is the channel gain of link $k$. By $\sigma_k^2$ we denote the noise spectral density (in Watt/Hz) in the receiver of service $k$. 
By interpreting $\vw$ as the probability of generating interference from the transmitter of a link to the receiver of another link (on any resource block) \cite{mogensen2007lte,ho2015power}, the \ac{SINR} that link $k$ experiences can be approximated by 
\vspace{-.5ex}  
\begin{equation}
\sinr_k(\vw) \approx \frac{p_k v_{k,k}}{\sum\limits_{l\neq k} v_{k,l}p_l w_l + \sigma_k^2} = \frac{p_k}{\left[\tilde{\mV}\diag(\vp)\vw + \tilde{\ve{\sigma}}\right]_k},
\label{eqn:SINR}
\vspace{-.5ex} 
\end{equation}
where $\tilde{\mV}\in\RN^{K\times K}$ denotes the {\it interference coupling matrix} with the $(k,l)$-th entry defined as $v_{k,l}/v_{k,k}$, and $\tilde{\ve{\sigma}}:=[\sigma_1^2/v_{1,1}, \ldots, \sigma_K^2/v_{k,k}]\in\RP^K$. 

The achievable rate (in bit/s) of service $k$ is computed by 
\vspace{-1ex}
\begin{equation}
r_k(\vw)  = w_k\overline{W}\log_2(1 + \sinr_k(\vw)).
\label{eqn:DLrate}
\vspace{-1ex}
\end{equation} 
The objective of the first algorithm proposed here is to maximize the worst-case service-specific {\it \ac{QoS} satisfaction level}, defined as the ratio of the achievable rate $r_k(\vw)$ to the rate demand $\bar{r}_k$, subject to a per-\ac{BS} load constraint. Formally, the problem is stated as follows:
\vspace{-1ex}
\begin{subequations}\label{eqn:Maxmin_w}
\begin{align}
\maximize_{\vw\in\RN^K} \ & \min_{k\in\Ks} \  r_k(\vw)/\bar{r}_k \label{eqn:Maxmin_w_1}\\
\subject \ &  \|\mA\vw\|_{\infty}\leq \theta, \label{eqn:Maxmin_w_3}
\vspace{-1ex}
\end{align}
\end{subequations}
where $\|\cdot\|_{\infty}$ denotes the $L^{\infty}$-norm. Note that \eqref{eqn:Maxmin_w_3} implies a resource reuse factor of $1$, and with full load constraint we have $\theta = 1$; i.e., $(\forall n \in\Ns) \sum_{k\in\Ks} a_{n,k}w_k\leq 1$. 

In this work, we are interested not only in the solution to the above problem, but also in answering the following question:\\
{\it Is the load limit $\theta = 1$ with resource reuse factor $1$ an efficient operation point?}
\vspace{-1ex}
\subsection{Optimal Resource Allocation and Performance Limits}\label{subsec:DL_ResAllo}
To study the resource efficiency of the solution to Problem \eqref{eqn:Maxmin_w} as a function of the budget $\theta$, we start with next technical result. The proof is omitted because it is similar to that in \cite[Ex. 2]{cavalcante2014toward}.
\vspace{-.5ex}
\begin{lemma}
Suppose that the \ac{SINR} is modeled as in \eqref{eqn:SINR}, then the mapping $\ve{T}:=\left[T^{(1)}(\vw), \ldots, T^{(K)}(\vw)\right]$, where 
$(\forall k\in\Ks) \ T^{(k)}(\vw):\RN^K\to\RP: \vw\mapsto \bar{r}_k/\left(\overline{W}\log_2(1 + \sinr_k(\vw))\right)$, is an \ac{SIF}.
\label{lem:SIF_T_w}
\end{lemma}
\vspace{-1ex}

By using Lemma \ref{lem:SIF_T_w} and \cite[Prop. 1]{Liao17} we verify that the optimal solution $\vw^{\ast}$ to \eqref{eqn:Maxmin_w} satisfies $(\forall k\in\Ks) \ u_k(\vw^{\ast}) = c^{\ast}$ \cite[Prop. 1]{Liao17} for some $c^\star\in\RP$. This result can be alternatively obtained as follows. Assuming that each \ac{BS} serves at least one user and that every user is served by a \ac{BS}, which guarantees that each \ac{BS} serves a nonempty and unique set of users, we have that all rows of the assignment matrix $\mA$ are linearly independent. Therefore, $\mA$ is a nonnegative full (row) rank matrix, so  the function $g:\R^K\to\RN:\vw\mapsto\|\mA|\vw|\|_{\infty}$ (which in particular implies $g(\vx)=\|\mA\vw\|$ for $\vw\in\R_+^K$), where $|\cdot|$ denotes the coordinate-wise absolute value of a vector, is a monotone norm.   
Thus, the problem in \eqref{eqn:Maxmin_w} is an instance of that in \eqref{eqn:Maxmin}, and the solution to \eqref{eqn:Maxmin_w} can be easily obtained with the iterations in \eqref{eqn:FP_w} as explained in Section~\ref{sec:def}. Furthermore, by using Proposition \ref{prop:asymp}, we can compute the asymptotic mapping $\ve{T}_{\infty}$ associated with $\ve{T}$. By doing so, as shown below, we are able to compute $(\vw_{\infty}, \lambda_{\infty})$ defined in Proposition \ref{prop:T_infty_compute}, and the performance limits in Proposition \ref{prop:PerformLimits} become readily available.  
\begin{proposition}
\vspace{-1ex}
Let $\ve{T}:\RN^K\to\RP^K:\vw\mapsto\left[T^{(1)}(\vw), \ldots, T^{(K)}(\vw)\right]$ be as defined in Lemma \ref{lem:SIF_T_w}. Suppose that $\vp\in\RP^K$ and $\mVt$ is irreducible, implying that each service is interfered by at least another of the services. The asymptotic mapping $\ve{T}_{\infty}:\RN^K\to\RN^K:\vw\mapsto[T^{(1)}_{\infty}(\vw), \ldots, T^{(K)}_{\infty}(\vw)]$ associated with $\ve{T}$ is given by:  
\begin{align}
\vspace{-1ex}
\ve{T}_{\infty}(\vw) & = \diag(\vp)^{-1}\ma{\Phi}\mVt\diag(\vp)\vw \label{eqn:T_infty_w_1}\\
\mbox{where } \ma{\Phi} & :=\frac{\ln(2)}{\overline{W}} \diag\left(\bar{r}_1,\ldots, \bar{r}_K\right) \nonumber.
\vspace{-1ex}
\end{align}
Furthermore, there exists a unique positive solution $(\vw_{\infty},\lambda_{\infty})\in\RP^K\times\RP$ to the \ac{CEVP}
\vspace{-1ex}
\begin{equation}
\ve{T}_{\infty}(\vw_{\infty}) = \lambda_{\infty}\vw_{\infty}, \ \|\mA\vw_{\infty}\|_{\infty} = 1\label{eqn:condi_eigen},
\end{equation}
which is given by 
\begin{equation}
\lambda_{\infty} = \lambda_{\ast}, \ \vw_{\infty} = \vw_{\ast}/\|\mA \vw_{\ast}\|_{\infty}
\label{eqn:solutionToCondEigen},
\end{equation}
\label{prop:T_infty_compute}
where $\lambda_{\ast}$ and $\vw_{\ast}$ are, respectively, the unique largest real eigenvalue and a corresponding real eigenvector of the matrix $\mG :=\diag(\vp)^{-1}\ma{\Phi}\mVt\diag(\vp)$. 

\end{proposition}
\begin{proof}

Applying Proposition \ref{prop:asymp}, we have 
\begin{align*}
&(\forall k\in\Ks)(\forall\vw\in\RN^K) \ T^{(k)}_\infty = \lim_{h\to \infty} T^{(k)}(h\vw)/h\\
&= \lim_{h\to\infty} \bar{r}_k/\left(h\overline{W}\log_2\left(1+ p_k/\left[h\mVt\diag(\vp)\vw + \tilde{\ve{\sigma}}\right]_k\right)\right).
\end{align*} 
Defining $g(h):=\log_2\left(1+ p_k/\left[h\mVt\diag(\vp)\vw + \tilde{\ve{\sigma}}\right]_k\right)$ and $f(h):=\bar{r}_k/(\overline{W}h)$, we have that $\lim_{h\to\infty}g(h)= 0$, $\lim_{h\to\infty}f(h) = 0$, $g'(h)\neq 0$ for $h \in\R_+$, and that $\lim_{h\to \infty} f'(h)/g'(h)$ exists. By using {\it L'H\^{o}pital's rule}, we verify that $T^{(k)}_\infty = \ln(2)\bar{r}_k\left[\mVt\diag(p)\vw\right]_k/(\overline{W}p_k)$. We obtain \eqref{eqn:T_infty_w_1} by writing $\ve{T}_{\infty}:=\left[T_{\infty}^{(1)}, \ldots, T_{\infty}^{(K)}\right]$ in matrix form. 
%
We can also verify that $\mG$ is nonnegative and irreducible. As a result, by  Perron-Frobenius theory \cite[p. 673]{meyer2000matrix}, $\mG$ has a simple positive eigenvalue $\lambda_{\ast}\in\RP$ associated with a positive right eigenvector $\vw_{\ast}$. Furthermore, any other real eigenvalue $\lambda$ satisfies $|\lambda|\le \lambda_{\ast}$, and if $T_\infty(\vw) = \mG\vw = \lambda \vw$ for some $\vw\in\R_{+}^K\backslash\{\boldsymbol{0}\}$, then we have $\lambda = \lambda_\ast$ and $\vw = c\vw_\ast$ for some $c\in\R_{++}$. In particular, with $c=1/\|\mA\vw_{\ast}\|_{\infty}$, $\vw_\infty = c \vw_\ast$, and $\lambda_\infty = \lambda_{\ast}$, we deduce $T_\infty(\vw_{\infty}) = \mG \vw_{\infty}  = \lambda_{\infty}\vw_{\infty}$ and $\|\mA\vw_{\infty}\|_{\infty} = \left\|\mA\left(\vw_{\ast}/\|\mA\vw_{\ast}\|_{\infty}\right)\right\|_{\infty} = \|\mA\vw_{\ast}\|_{\infty}/\|\mA\vw_{\ast}\|_{\infty}  = 1$. The above implies that (\ref{eqn:solutionToCondEigen}) is the unique solution to the CEVP in  (\ref{eqn:condi_eigen}), and the proof is complete. 
\end{proof}
\vspace{-1.5ex}
\begin{Rem} Proposition \ref{prop:T_infty_compute} shows an efficient method to compute $(\vw_{\infty},\lambda_{\infty})$ as the solution to the \ac{CEVP} \eqref{eqn:condi_eigen}. Note that, \emph{in this specific case,  the asymptotic mapping associated with a nonlinear mapping (defined in Lemma \ref{lem:SIF_T_w}) becomes linear}, as shown in \eqref{eqn:condi_eigen}.
Moreover, to compute the eigenvalues of $\mG = \diag(\vp)^{-1}\ma{\Phi}\mVt\diag(\vp)$, we can equivalently compute the  eigenvalues of the matrix $\ma{\Phi}\mVt$, which is independent of $\vp$\footnote{Suppose that $\mD\in\R^{K\times K}$ is an invertible matrix  and $\mX\in\R^{K\times K}$,  the eigenvalues of the matrices $\mX$ and $\mD\mX\mD^{-1}$ are the same \cite[Rem. 1]{cavalcante2016elementary}.}. 
%
\end{Rem}
\vspace{-1ex}
\subsection{Partial Resource Muting}\label{subsec:DL_muting}
%
Having $(\vw_{\infty}, \lambda_{\infty})$ in hand, we are now able to compute $\theta^{\trans}$ with \eqref{eqn:TransPoint} and answer the following question related to the resource efficiency raised in Section~\ref{subsec:DL_SysModel}:\\ 
{\it If the transition point yields $\theta^{\trans}<1$,   full resource reuse (i.e., $\theta=1$) may not be an efficient operation point because we are likely operating in an interference limited region where the resource availability $\theta$ can be decreased without noticeable changes in utility -- see Fig~\ref{fig:PowerUtility_Bound}.}

The new challenge arises:\\
{\it How to improve the resource efficiency if the network is operating in an interference-limited inefficient region?} 

Since the bottleneck users usually consume most of the resources and impair the performance, we consider muting partial resources in the neighboring cells to mitigate the interference received in (and generated by) the bottleneck users. Based on $(\lambda_{\infty}, \vw_{\infty})$ obtained  in \eqref{eqn:solutionToCondEigen} and the derived transition point $\theta^{\trans}$, we propose a resource muting scheme consisting of the following steps.
%
%
\subsubsection{Triggering the Resource Muting Scheme}
$\theta^{\trans}<1$ implies inefficient usage of resources in the region $[\theta^{\trans}, 1]$ due to heavy interference. Instead of allocating all resources in at least one \ac{BS} to achieve only a slight increase in utility, the network may better benefit from muting partial resources to reduce the interference of bottleneck users. Therefore, the resource muting scheme can be triggered if $\theta^{\trans}<1$. 
\subsubsection{Modifying the Interference Pattern and Load Constraints}
Suppose that a set of bottleneck users, denoted by $\Ks^{(\text{b})}$, is selected (details about the selection of bottleneck users are given in the next subsection).  The motivation is to allocate $\Ks^{(\text{b})}$ to a muting region, such that each user $k\in\Ks^{(\text{b})}$ neither receives interference from the neighboring cells nor generates interference to those cells, as shown in Fig. \ref{fig:ResourceMutingDL}. 
\begin{figure}
\centering
		    \includegraphics[width=.8\columnwidth]{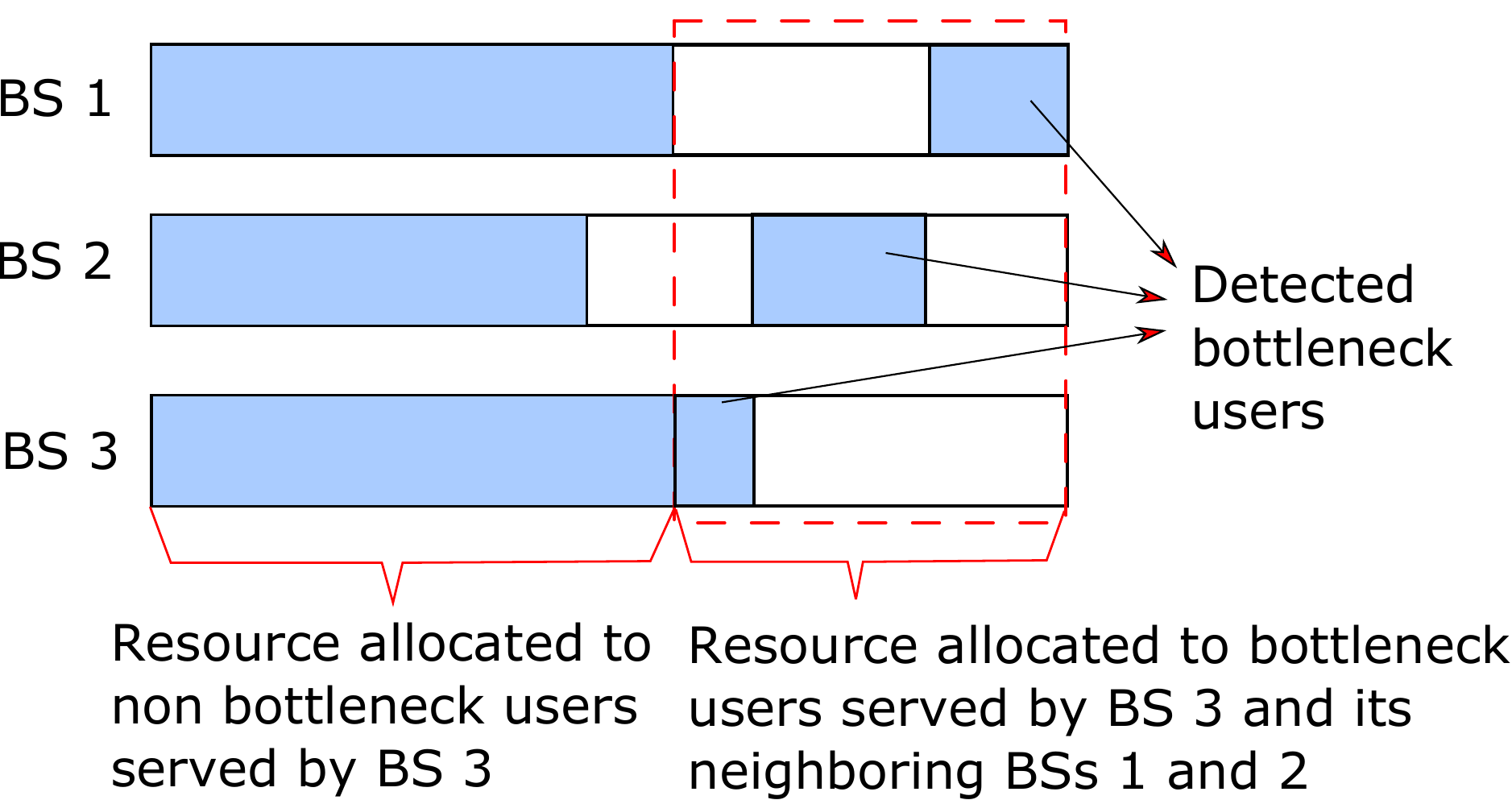}
\caption{Resource muting region in downlink.}
\label{fig:ResourceMutingDL}
\vspace{-3ex}
\end{figure}
Since the received/generated interference of the bottleneck users from/to the neighboring cells is canceled, the interference coupling matrix is updated as follows
\begin{equation}
\vspace{-1ex}
 (\forall k\in\Ks^{(\text{b})})(\forall m\in\Ns_{k})(\forall l\in\Ks_{m}) v_{k,l} = v_{l,k} = 0,
\label{eqn:updateInterference_2}
\end{equation}
where $\Ns_k$ denotes the set of neighboring cells of cell $k$.

On the other hand, to incorporate the muting region in the constraint, we modify the load constraint in \eqref{eqn:Maxmin_w_3} as shown below:
\begin{subequations}\label{eqn:modified_load_constaint}
\begin{align}
g(\vw) & \leq 1, \label{eqn:modified_load_constraint_1}\\
g:\R^K & \to\RN:\vw \mapsto \max_{n\in\Ns} \left( \sum_{k\in\Ks_n} |w_k| + \sum_{m\in\Ns_n}\sum_{l\in\Ks^{(\text{b})}_m}|w_l| \right)\label{eqn:modified_load_constaint_2},
\vspace{-2ex}
\end{align}
\end{subequations}
where $\Ks^{(\text{b})}_m$ denotes the set of bottleneck users in cell $m$. Constraint \eqref{eqn:modified_load_constaint} implies that, for each cell $n$, the sum of resources allocated to its serving users $\Ks_n$ (including both the bottleneck and non bottleneck users) and the resources allocated to all the bottleneck users served by its neighboring cells $m\in\Ns_n$ is limited to the total amount of resources. Note that the modified function $g$ is a monotone norm. Thus, the problem with the modified constraint \eqref{eqn:modified_load_constaint} is still an instance of that in \eqref{eqn:Maxmin}. 

\subsubsection{Detecting Bottleneck Users}\label{subsubsec:BottleneckDetect}
The asymptotic behavior $(\lambda_{\infty}, \vw_{\infty})$ provides a good guide to detect efficiently the bottleneck services, because it indicates the existing limits of the utility and the fraction of allocated resources as $\theta\to \infty$. More precisely, let the $k$th entry of $\vw_{\infty}$ be denoted by $w_{\infty}^{(k)}$. The larger the value of $w_{\infty}^{(k)}$, the higher the chance that the corresponding user impairs the system performance owing to the large amount of occupied resources, and, consequently, the higher the possibility that this user causes heavy interference. 

To show how $\vw_{\infty}$ can be useful in the  development of efficient heuristics for selecting the set of bottleneck users, we compare the following two approaches: {\it exhaustive search} and {\it successive selection}. As we discuss below, the latter approach is suboptimal, but its computational complexity is subtantially smaller than that of the former approach.
\begin{itemize}
\item {\it Exhaustive search}. Given a set of candidate users, denoted by $\Ks^{(c)}$, we find an optimal subset $\Ks^{(b)}\subseteq\Ks^{(c)}$ that provides the maximum utility. This problem poses a serious computational challenge: the number of subsets is exponential in the size of the domain. It might require an exhaustive search over $\sum_{i=1,\ldots, |\Ks^{(c)}|} {|\Ks^{(c)}| \choose i}$ possible subsets.
%
\item {\it Successive selection}: By sequentially selecting users with the highest values $w_{\infty}^{(k)}$, the resulting optimized utility has a general trend of first increasing and then decreasing, as shown in Fig. \ref{fig:bottleneck}.  Therefore, instead of exhaustively searching for the best set of bottleneck users, we propose an efficient heuristic to approximate the solution. We successively add the users with the next highest value of $w^{(k)}_{\infty}$ until the utility does not increase. 
\end{itemize}
Fig. \ref{fig:bottleneck} illustrates the motivation for {\it successive selection}. It shows the optimized utility depending on the number of  bottleneck users for four simulation instances, each of them with a different distribution for the user locations and traffic demands. Suppose that the selection of bottleneck users is based on the rank order the components of $\vw_{\infty}$. Let us sort the entries in $\vw_{\infty}$ in descending order, with the order of $w_{\infty}^{(k)}$ denoted by $o_k$. For example, if the number of bottleneck users is $3$, then, the highest ranked users, i.e., $\{k \in \Ks: o_k\in\{1,2,3\}\}$ are selected. The curves generated by different instances show a common trend of utility increase when the users with the highest values of $w_{\infty}^{(k)}$ are allocated in the muting region. Then, if too many bottleneck users are selected, the utility decreases. Such pattern reflects the tradeoff between resource utilization and interference mitigation. 
Therefore, although the curves are not always piecewise monotonic, numerical results show that adding users successively until the utility starts to decrease provides a simple means of identifying bottleneck users.
\subsubsection{Fixed Point Iteration}\label{subsubsec:FP_iteration_DL}
Once the set of bottleneck users is updated, the load constraint is modified as shown in \eqref{eqn:modified_load_constaint} to incorporate the muting scheme, and the fixed point iteration \eqref{eqn:FP_w} (with $\theta = 1$ and $\|\cdot\| = g(\cdot)$ as defined in \eqref{eqn:modified_load_constaint}) can be applied to solve Problem \eqref{eqn:Maxmin}. 
 \begin{figure}[t]
    \centering
		    \includegraphics[width=1\columnwidth]{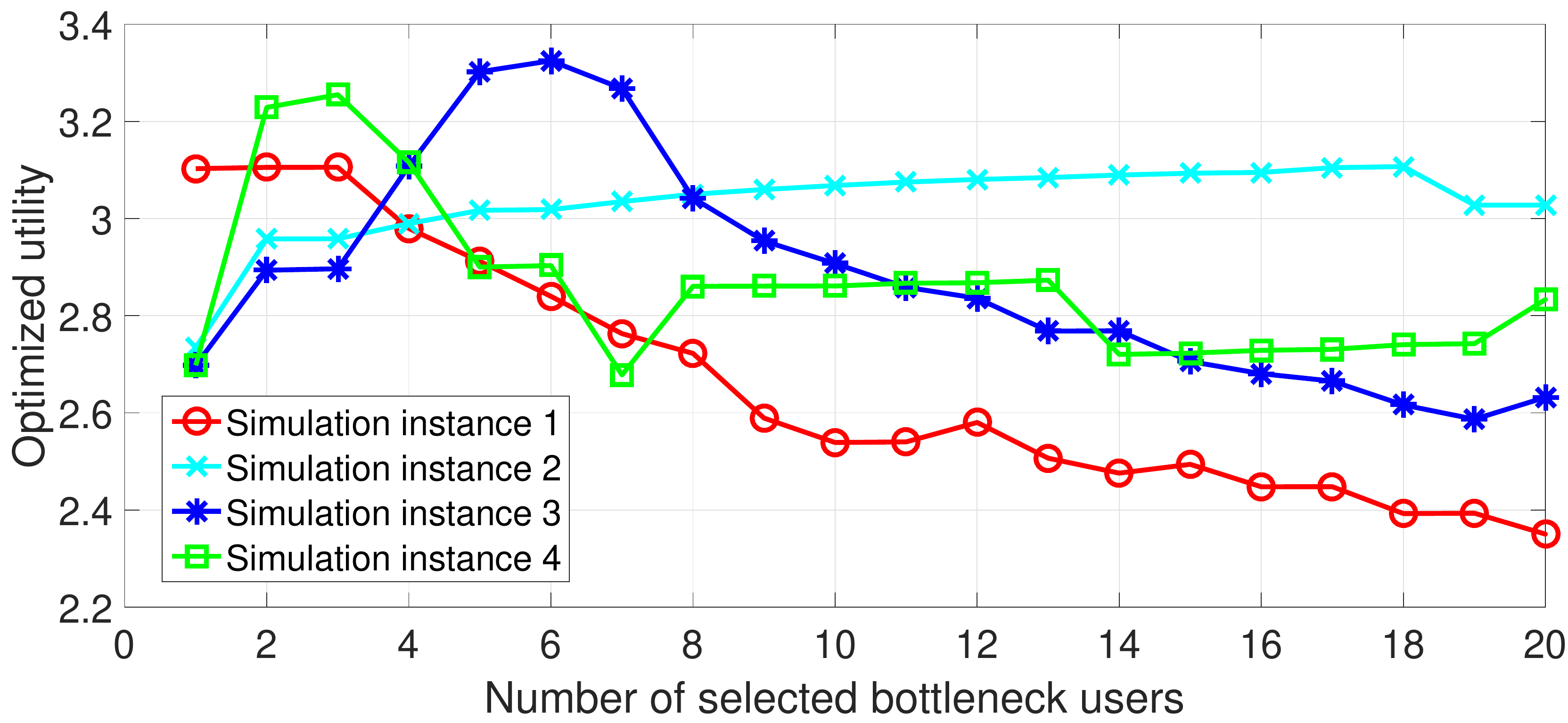}
        \caption{Example of successive selection of bottleneck users.}
 \label{fig:bottleneck}
\vspace{-2ex}
\end{figure}
Algorithm \ref{algo:AlgorithmDetectBottleneck} summarizes the proposed resource muting mechanism. 
\begin{algorithm}[t]
\vspace{-.5ex}
\SetKwInOut{Input}{input}
\SetKwInOut{Output}{output}
\caption{Partial Resource Muting}
\label{algo:AlgorithmDetectBottleneck}
\Input{$n\leftarrow 0$, $\Ks^{(\text{b})}(0) =\emptyset$
}
\Output{$\Ks^{(\text{b})}_{\ast}$, $\vw^{\ast}$}
Compute $\vw^{\ast}(0), c^{\ast}(0)$ using \eqref{eqn:FP_w} with $\theta=1$\;
Compute $\vw_{\infty}, \lambda_{\infty}$ and transition point $\theta^{\trans}$ using Proposition \ref{prop:PerformLimits} and \ref{prop:T_infty_compute}\;
\If{$\theta^{\trans}< 1$}{
Sort $w^{(k)}_{\infty}$ in descending order\;
\While{$n = 0$ or $c^{\ast}(n) - c^{\ast}(n-1)\geq 0$}{
$\Ks^{(\text{b})}_{\ast}\leftarrow \Ks^{(\text{b})}_{\ast}(n)$, $\vw^{\ast}\leftarrow\vw^{\ast}(n)$\;
$n\leftarrow n+1$\;
$\Ks^{(\text{b})}(n)\leftarrow \Ks^{(\text{b})}(n-1)\cup\{k|o_k = n\}$\;
Update interference pattern and load constraints with \eqref{eqn:updateInterference_2} and \eqref{eqn:modified_load_constaint}\;
Compute $c^{\ast}(n), \vw^{\ast}(n)$ corresponding to $\Ks^{(\text{b})}(n)$ with fixed point iteration \eqref{eqn:FP_w}\;
}
}\ElseIf{$\theta^{\trans}\geq 1$}{
$\Ks^{(\text{b})}_{\ast}\leftarrow \Ks^{(\text{b})}_{\ast}(0)$, $\vw^{\ast}\leftarrow\vw^{\ast}(0)$
}
\BlankLine
\vspace{-1ex}
\end{algorithm}
\subsection{Numerical Results}\label{subsec:DL_Numerical}
We consider a real-world scenario with 15 three-sector macro \acp{BS} and 10 pico \acp{BS} in the city center of Berlin, Germany. The locations of the macro \acp{BS} are given by the real data set \cite{momentum}, while the pico \acp{BS} are placed uniformly at random in the playground. 
The  macro \acp{BS} are equipped with directional antennae with transmit power of $43$dBm, while the pico \acp{BS} have omnidirectional antennae with transmit power of $30$dBm. The total bandwidth is 10 MHz. The macrocell pathloss is obtained from the real data set \cite{momentum}, and the picocell pathloss uses the 3GPP LTE model in \cite{3GPP36814}. Uncorrelated fast fading characterized by Rayleigh distribution is implemented on top of the pathloss. A fixed number of users are randomly and uniformly distributed on the playground. The traffic demand per user is uniformly distributed between $[0,10]$ MBit/s with an average value of $5$ MBit/s.  
%
\subsubsection{Resource Efficiency with or without Resource Muting}\label{subsubsec:simu_ResourceEff_Muting}
Fig. \ref{fig:Comp_NonMuting} illustrates how the transition point of the utility and the resource efficiency change when the resource muting scheme is activated. Fig. \ref{fig:MutingRegion_U} shows that the resource efficient region increases (or, equivalently, the interference-limited region decreases) when applying resource muting. The achieved utility $U(\theta)$ and resource efficiency $E(\theta)$  with the muting scheme are significantly higher at $\theta = 1$ (with the load constraint \eqref{eqn:modified_load_constaint}).
\subsubsection{Efficient Selection of Bottleneck Users}\label{subsubsec:simu_bottleneck}
To show that $\vw_{\infty}$ is useful to identify bottleneck users, we use Monte Carlo simulations to generate random locations and demands of users, and we compare the number of bottleneck users $\left|\Ks^{(\text{b})}\right|$ selected by {\it exhaustive search} and by {\it successive selection} as described in Section \ref{subsubsec:BottleneckDetect}. In the upper subfigure of Fig. \ref{fig:Exhaustive_vs_successive}, $500$ points of $\left(\left|\Ks^{(\text{b})}_{(\text{exh})}\right|, \left|\Ks^{(\text{b})}_{(\text{suc})}\right|\right)$ are plotted (note that there are overlapping points), with each corresponding to a distribution of user locations and traffic demands, where $\Ks^{(\text{b})}_{(\text{exh})}$ and $\Ks^{(\text{b})}_{(\text{suc})}$ denote the set of bottleneck users selected by {\it exhaustive search} and by {\it successive selection}, respectively. The lower subfigure shows the empirical probability density of the number of bottleneck users. Although Fig. \ref{fig:Exhaustive_vs_successive} shows that {\it successive selection} generally selects less members than {\it exhaustive search}, it is shown in Fig. \ref{fig:DL_MutingCompare} that {\it successive selection} achieves similar performance to  {\it exhaustive search}. 
%
%
%
\subsubsection{Performance Improvement}\label{subsubsec:PerformImprove}
In Fig. \ref{fig:DL_MutingCompare} we compare the performance of four protocols for $K = 200$ users and $K = 400$ users. The performance is characterized by the \ac{CDF} of the achievable utility derived from $1000$ Monte Carlo simulations with random distribution of user locations and traffic demands. The four protocols are described below:
\begin{itemize}
\item[(\Rmnum{1})] {\it Non-muting}:  We solve Problem \eqref{eqn:Maxmin_w} with $\theta=1$ by applying the fixed point iteration in  \eqref{eqn:FP_w}.  
\item[(\Rmnum{2})] {\it Muting based on $\ve{w}_{\infty}$, successive selection}: We solve Problem  \eqref{eqn:Maxmin_w} with Algorithm \ref{algo:AlgorithmDetectBottleneck}. The interference coupling and the load constraints are modified according to \eqref{eqn:updateInterference_2} and \eqref{eqn:modified_load_constaint}, respectively.
\item[(\Rmnum{3})] {\it Muting based on $\vw_{\infty}$, exhaustive search}: Similar to Algorithm \ref{algo:AlgorithmDetectBottleneck}, but the bottleneck users are selected by exhaustive search. 
\item[(\Rmnum{4})] {\it Muting based on $\ve{I}$, successive search}: Similar to Algorithm \ref{algo:AlgorithmDetectBottleneck}, but the bottleneck users are selected successively based on a different indicator $I_k:=\sum_{l\neq k} \left(p_lv_{k,l}w_l^{\ast} + p_k v_{l,k} w_k^{\ast}\right)$, which reflects the sum of the received and generated interference by user $k$.  
\end{itemize}
Comparing the performance of Protocol \Rmnum{1} with the other muting schemes in Fig. \ref{fig:DL_MutingCompare}, we observe that the muting schemes significantly improve the desired utility. Comparing Protocol \Rmnum{2} and \Rmnum{3}, we verify that successive selection performs only slightly worse than exhaustive search, but the computational effort is significantly reduced.  Furthermore,  comparing Protocol \Rmnum{3} and \Rmnum{4}, we verify that $\vw_{\infty}$ is a more appropriate indicator for bottleneck selection than the indicator based on interference measurements.  
%
%
\begin{figure}
    \centering
    \begin{subfigure}[t]{1\columnwidth}
        \includegraphics[width=\textwidth]{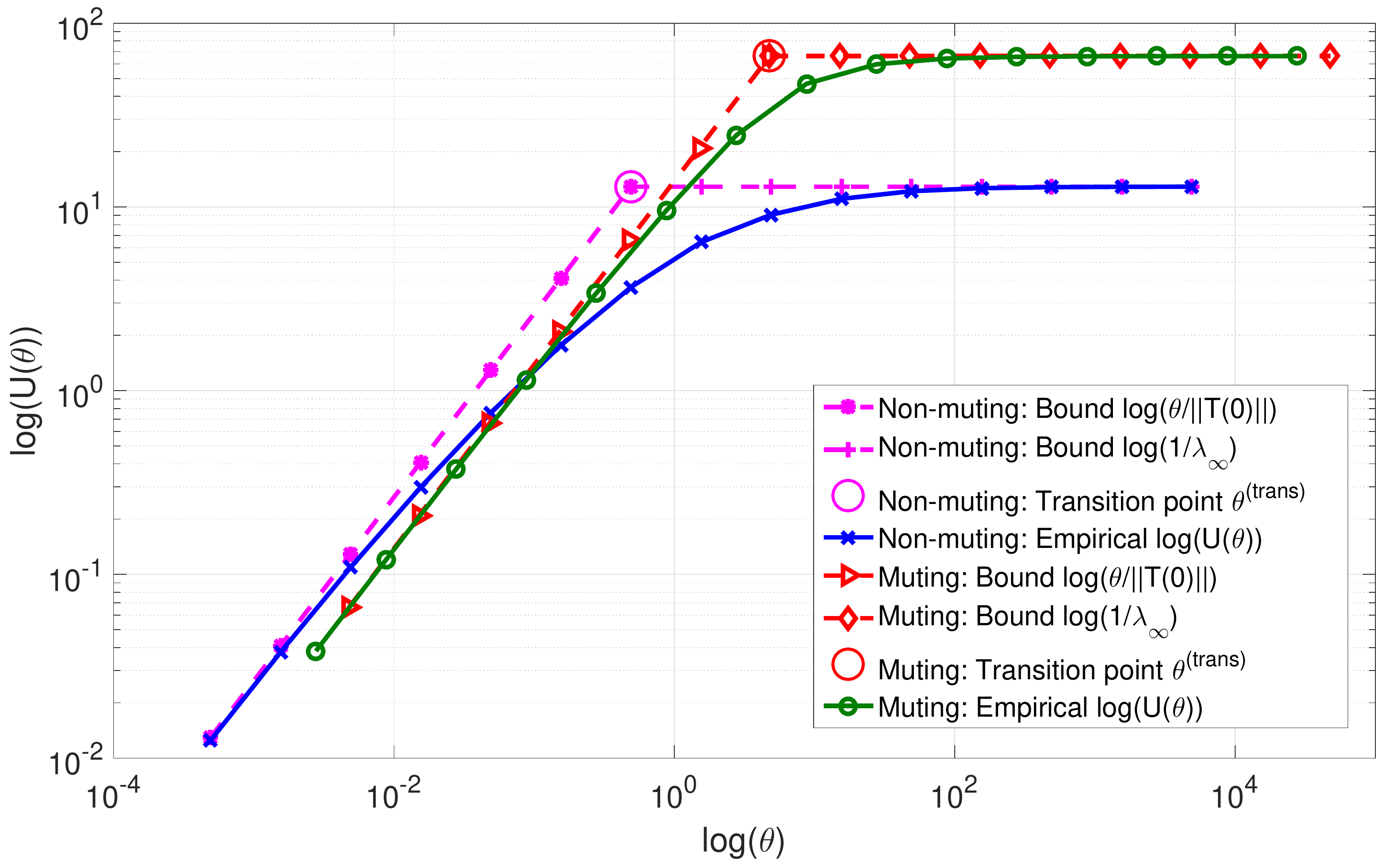}
        \caption{Log-log plot of utility depending on $\theta$.}
        \label{fig:MutingRegion_U}
    \end{subfigure}
    \begin{subfigure}[t]{1\columnwidth}
        \includegraphics[width=\textwidth]{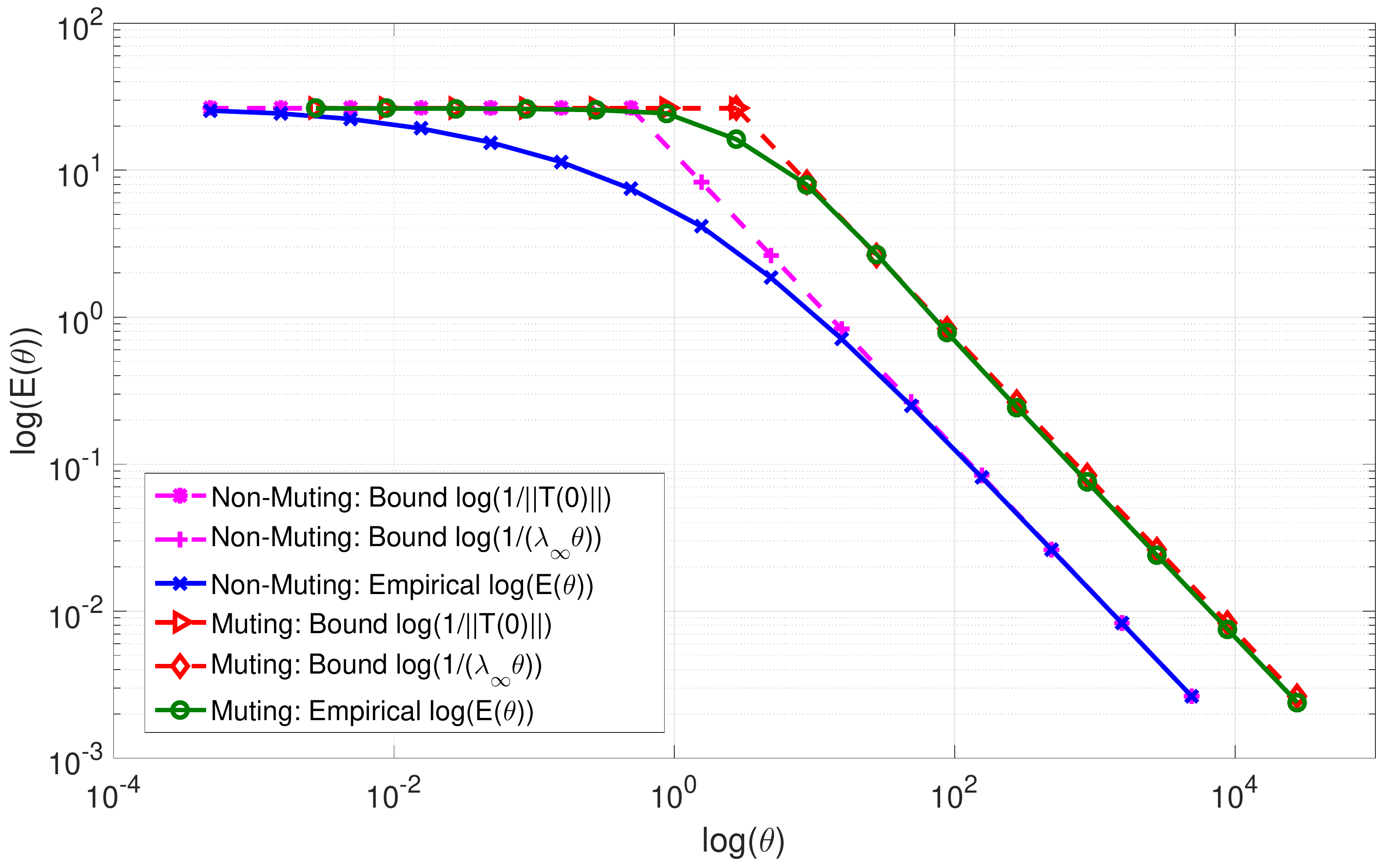}
        \caption{Log-log plot of resource efficiency depending on $\theta$.}
        \label{fig:MutingRegion_E}
    \end{subfigure}
    \caption{Comparison between the performance limits of non-muting scheme and muting scheme, $K= 200$.}
		\label{fig:Comp_NonMuting}
		\vspace{-4ex}
\end{figure}
%
%
%
%
 \begin{figure}[t]
    \centering
		    \includegraphics[width=1\columnwidth]{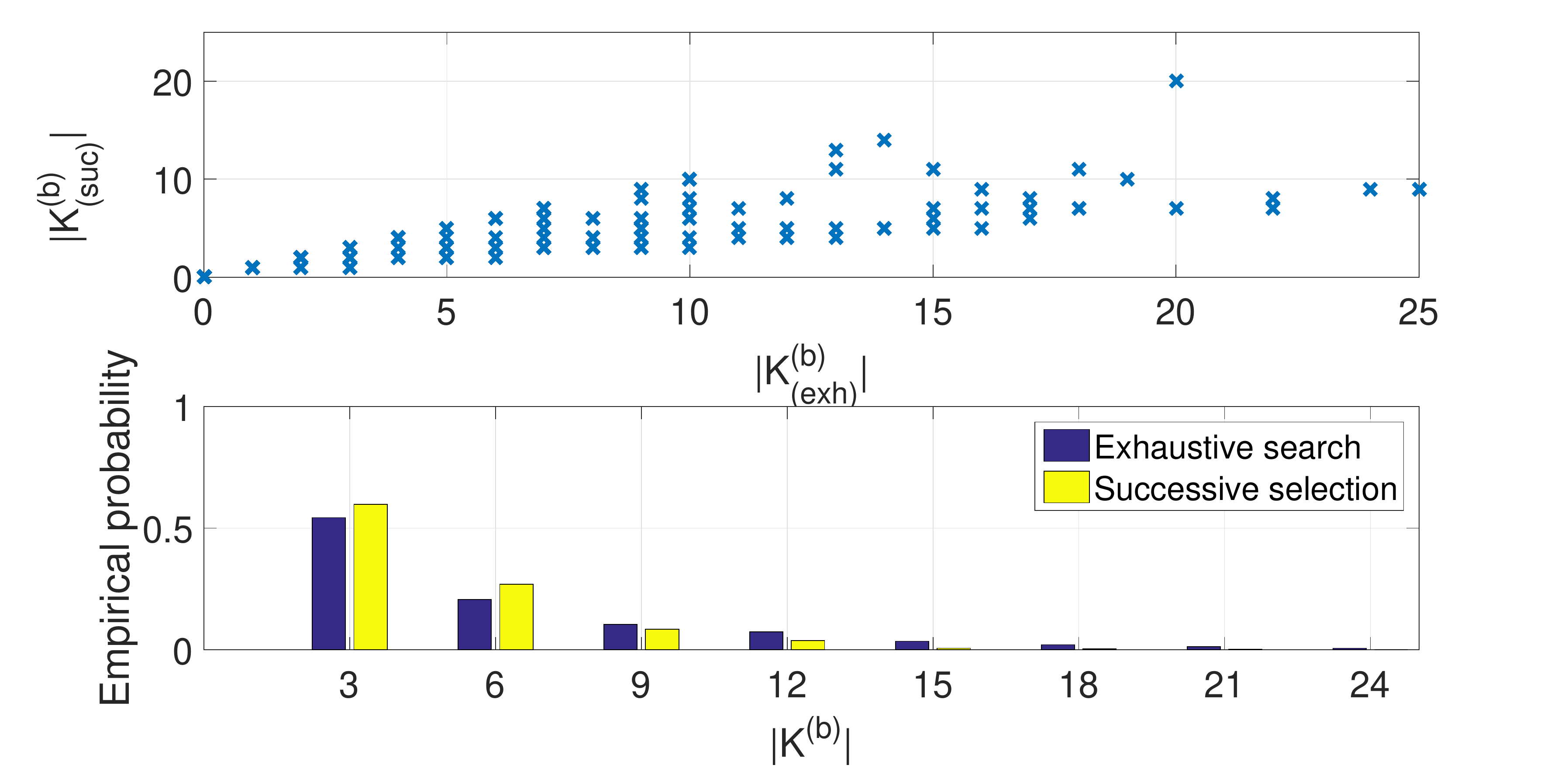}
								\vspace{-1ex}
        \caption{Number of bottleneck users selected by exhaustive searching and by successive selection.}
				\vspace{-3ex}
 \label{fig:Exhaustive_vs_successive}
\end{figure}
 \begin{figure}[t]
    \centering
		    \includegraphics[width=1\columnwidth]{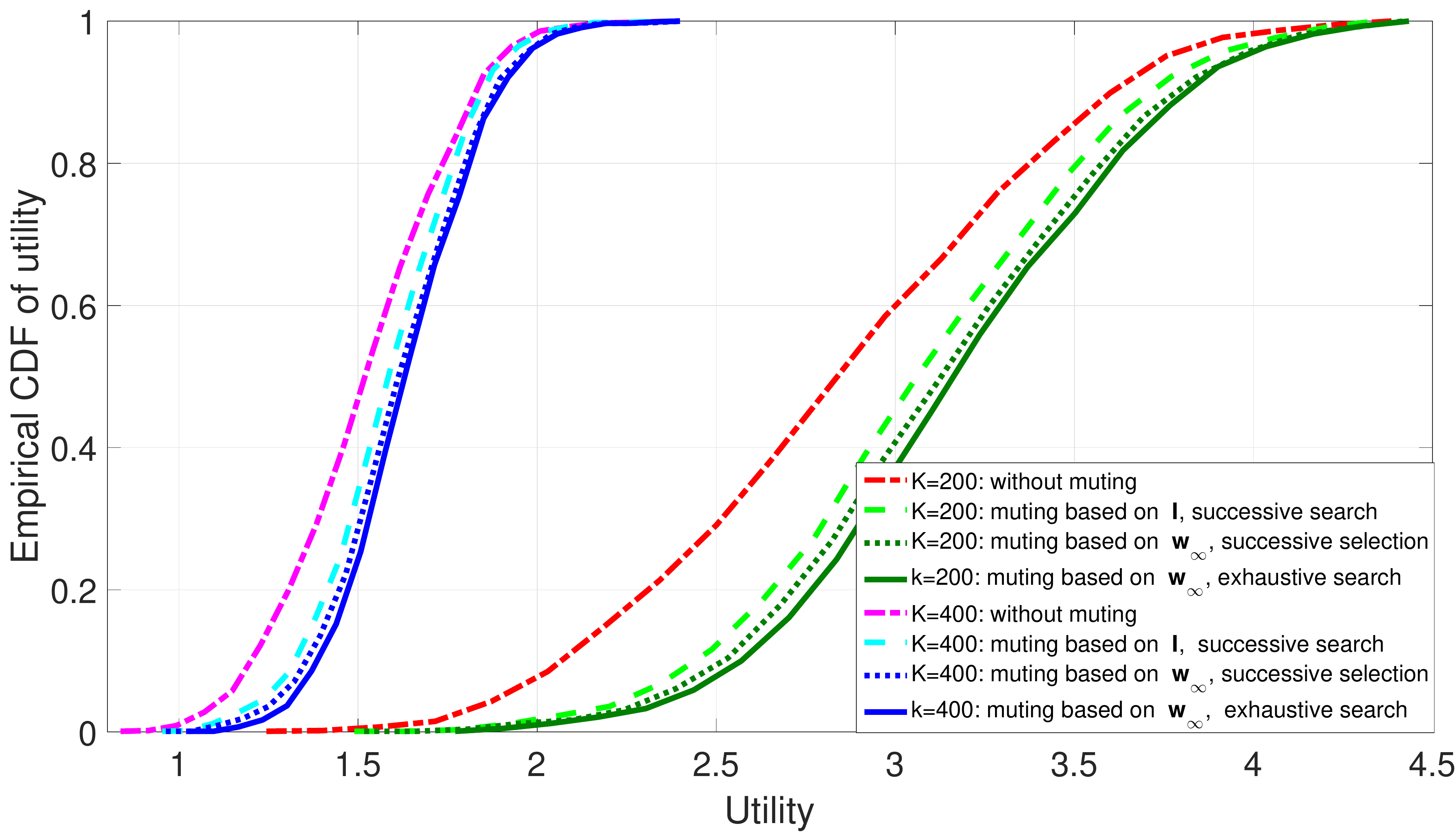}
				\vspace{-1ex}
        \caption{Performance comparison between max-min utility fairness with and without muting for $K = 200$ and $400$.}
				\vspace{-3ex}
 \label{fig:DL_MutingCompare}
\end{figure}
%

%
%
%
	%

%
%
%
\section{Extension of the Framework to 5G Networks}\label{sec:Extension_5G}
The framework can be extended to optimize 5G networks equipped with flexible duplex mechanisms. Flexible duplex is one of the key technologies in 5G to adapt to asymmetric \ac{UL} and \ac{DL} traffic with flexible resource allocation in joint time-frequency domain. The main challenge is the newly introduced \ac{IMI} between \ac{UL} and \ac{DL} as shown in Fig.\ref{fig:ULDLOverlapping}, which makes Problem \eqref{eqn:Maxmin_w} far more complex. In previous work \cite{Liao17} we have shown how to incorporate \ac{IMI} into the interference model, and we have developed a novel algorithm called \ac{SAFP} to approximate the solution to \eqref{eqn:Maxmin_w}. In this study we briefly describe the interference model and the \ac{SAFP} algorithm, and we put more focus on a new resource muting scheme based on the spectral properties of asymptotic mappings.
\vspace{-.5ex}
\subsection{Joint UL/DL System Model}
We now generalize the \ac{DL} system model defined in Section \ref{subsec:DL_SysModel} by considering a set of services $\Ks$, including both \ac{UL} and \ac{DL} services. As in previous sections, let the \ac{BS}-to-service assignment be denoted by $\mA\in\{0,1\}^{N\times K}$. Without resource muting, the utility max-min fairness problem can also be written in the general form \eqref{eqn:Maxmin_w}. However, the \ac{SINR} model needs to be modified to take \ac{IMI} into account. As shown in Fig. \ref{fig:ULDLOverlapping}, inter-cell interference appears in the overlapping resource region. We introduce the overlapping factors $\left(c_{k,l}\right)$, collected in $\mC(\vw)\in\RP^{K\times K}$, to incorporate the probability that intra- and inter-mode  interference appear for a given resource allocation $\vw$. Let $\ve{\nu}^{(x)}:=\left[\nu_1^{(x)},\ldots, \nu_N^{(x)}\right]$, $x\in\{\text{u}, \text{d}\}$ denote the \ac{UL} or \ac{DL} load (i.e., the fraction of occupied resources) of all cells that can be computed with $\vw$. For example, $\nu_n^{\ul}:=\sum_{k\in\Ks^{\ul}_n} w_k$ denotes the load of \ac{BS} $n$ in \ac{UL}, where $\Ks^{\ul}_n\subseteq\Ks_n$ is the set of \ac{UL} services in \ac{BS} $n$. The \ac{SINR} of service $k$ incorporating \ac{IMI} is approximated by
\begin{align}
\sinr_k(\vw) & \approx \dfrac{p_k}{\left[\left(\mC(\vw)\circ\mVt\right)\diag(\vp)\vw + \tilde{\ve{\sigma}}\right]_k} \label{eqn:SINR_withC}\\
\mbox{with } \mC(\vw) & := \left(c_{k,l}\right) \in\RN^{K\times K}, \label{eqn:MatrixC_0} \\
c_{k,l} & := 
\begin{cases}
 \left[ \left(\nu_{n_l}^{(x_l)} + \nu_{n_k}^{(x_k)}-1\right)/\nu_{n_k}^{(x_k)}\right]^+ & \mbox{  if } x_l \neq x_k \\
 \min\left\{1, \nu_{n_l}^{(x_l)}/\nu_{n_k}^{(x_k)}\right\} & \mbox{ if } x_l = x_k,\nonumber
\end{cases}
\end{align}  
where $\circ$ denotes the Hadamard product, $x_k\in\{\text{u}, \text{d}\}$ denotes the duplex mode of service $k$,  and $n_k$ denotes the serving \ac{BS} of $k$. The approximation is based on two probabilities:
\begin{itemize}
\item The overlapping factor $c_{k,l}$ is the ratio of the overlapping area between load $\nu_{n_l}^{(x_l)}$ and $\nu_{n_k}^{(x_k)}$ to the load $\nu_{n_k}^{(x_k)}$, which roughly indicates the probability that an arbitrary resource unit allocated to mode $x_l$ in \ac{BS} $n_l$ causes interference to \ac{BS} $n_k$ in mode $x_k$. Note that the first case $x_l\neq x_k$ corresponds to the \ac{IMI} between \ac{UL} and \ac{DL}, while $x_l = x_k$ corresponds to the intra-mode interference.
\item $w_l$ collected in $\vw$ serves as the probability that an arbitrary resource unit in $n_l$ is allocated to service $l$. 
\end{itemize} 
The multiplication of $c_{k,l}$ by $w_l$ loosely approximates the probability that a resource unit allocated to $l$ in duplex mode $x_l$ in \ac{BS} $n_l$ causes interference to service $k$ in duplex mode $x_k$ in \ac{BS} $n_k$.  
\subsection{Successive Approximation of Fixed Point}
Introducing $\mC$ into \eqref{eqn:MatrixC_0} removes the properties of $\ve{T}$ given in Lemma \ref{lem:SIF_T_w}, which further leads to possibly more than one fixed point of the resulting \ac{CEVP} \eqref{eqn:cevp}.  In \cite{Liao17} we developed a novel algorithm \ac{SAFP} to approximate the near-optimal fixed point of the \ac{CEVP}. The novel proposed algorithm assisted with {\it random initialization} and {\it successive approximation} is summarized below:
\begin{itemize}
\item The algorithm runs for $Z^{(\tmax)}$ times, where at the $i$-th time  different random initializations of $\vw^{(0)}:=\hat{\vw}_i$ and the corresponding $\hat{\mC}:=\mC(\hat{\vw}_i)$ are used.  
\item For each initialization, we iteratively perform the following two steps: 1) we use the fixed point iteration \eqref{eqn:FP_w} with respect to the approximated $\hat{\mC}$ and derive $\vw^{(t)}$, and 2) update $\hat{\mC} = \mC(\vw^{(t)})$ and increment $t$. The iteration stops if $\|\vw^{(t)}-\vw^{(t-1)}\|\leq \epsilon$, where $\epsilon$ is a distance threshold.
\item Each random initialization converges to a fixed point (not necessarily different from those derived from other initializations). We choose the solution with the maximum utility.
\end{itemize}
As shown in \cite{Liao17}, the algorithm converges for each random initialization. With a limited number of random initializations, the algorithm is able to find a solution with low computational complexity. 
\subsection{Partial Resource Muting}
We proposed partial resource muting scheme for 4G \ac{DL} wireless networks in Section \ref{subsec:DL_muting} to improve the resource efficiency. Along similar lines, the resource muting scheme can be tailored for \ac{IMI} mitigation in 5G networks enabling flexible duplex. 
To deal with the complex interference caused by \ac{IMI}, we have introduced \ac{SAFP} as the tool to  find efficiently the approximated solution to the \ac{CEVP} with the existence of possibly multiple fixed points. Hence, by replacing the fixed point iteration with \ac{SAFP} (line $10$ in Algorithm \ref{algo:AlgorithmDetectBottleneck}), the same steps proposed in Algorithm \ref{algo:AlgorithmDetectBottleneck} can be applied to optimize the joint \ac{UL}/\ac{DL} resource allocation, particularly the set of bottleneck users allocated to the muting region and the resource fraction allocated to them. 

It is worth mentioning that the steps used for triggering the muting scheme and for detecting the bottleneck users require the computation of the transition point $\theta^{\trans}$ and of the tuple $(\lambda_{\infty}, \vw_{\infty})$. However, with the modified \ac{SINR} model  \eqref{eqn:SINR_withC}, $\ve{T}(\vw)$ is not an \ac{SIF}, and Proposition \ref{prop:T_infty_compute} cannot be applied in a straightforward manner to compute the desired values. Therefore, we propose to approximate $\ve{T}(\vw)$ by replacing $\mC(\vw)$ with the converged approximation of $\hat{\mC}$ achieved by \ac{SAFP}. Let the approximation of $\ve{T}(\vw)$ be denoted by $\ve{T}_{\hat{\mC}}(\vw)$. Since $\hat{\mC}$ is  a matrix with known entries, $\ve{T}_{\hat{\mC}}(\vw)$ is an \ac{SIF}. Hence, we can use Proposition \ref{prop:asymp} to compute $\ve{T}_{\hat{\mC}, \infty}(\vw)$, and further derive $(\lambda_{\hat{\mC}, \infty}, \vw_{\hat{\mC}, \infty})$ with respect to $\ve{T}_{\hat{\mC}, \infty}(\vw)$ along similar lines to Proposition \ref{prop:T_infty_compute}. Using the same technique described in Section \ref{subsec:DL_muting}, we obtain the triggering parameter $\theta^{\trans}$ and the set of bottleneck users from $\left(\lambda_{\hat{\mC}, \infty}, \vw_{\hat{\mC}, \infty}\right)$. 
%
\begin{figure}[t]
\centering
\includegraphics[width=.75\columnwidth]{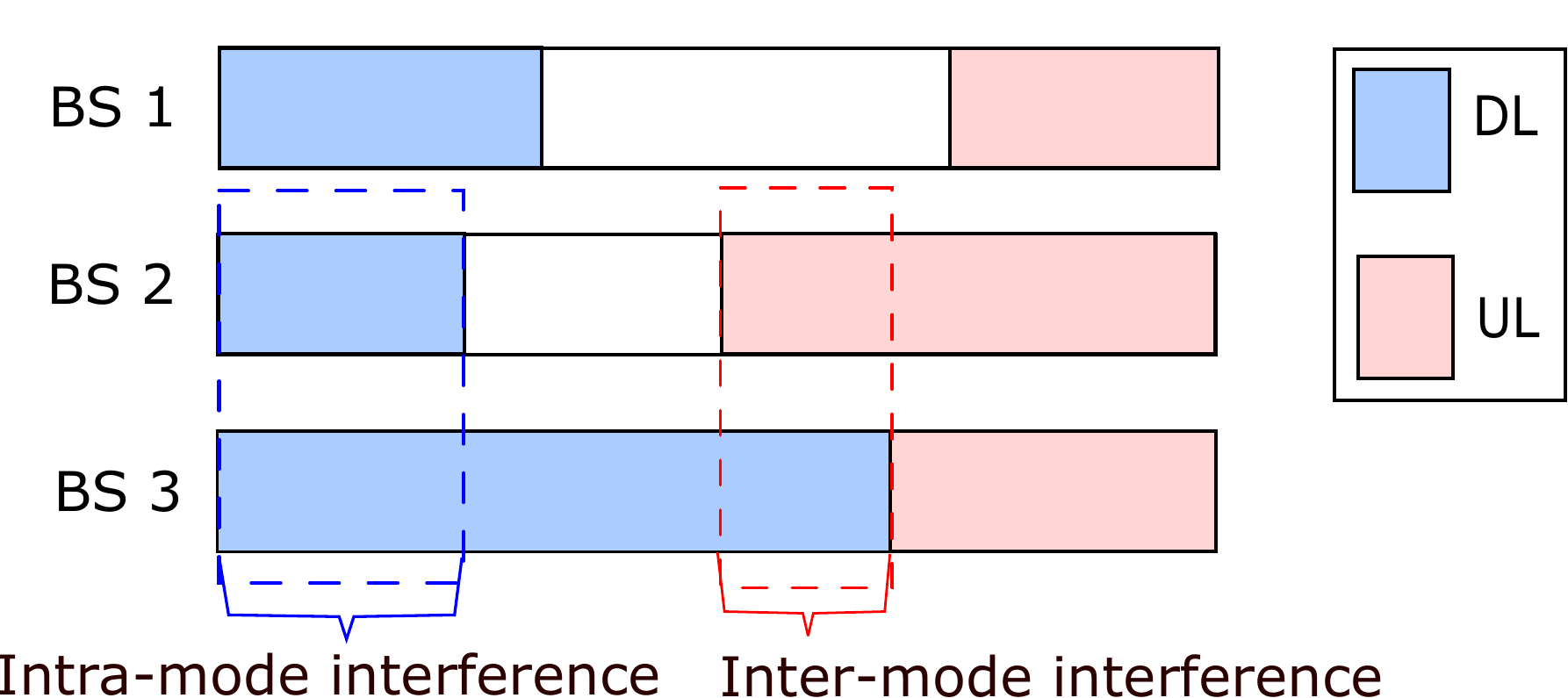}
\caption{Regions where inter-cell interference appears.}
\vspace{-2ex}
\label{fig:ULDLOverlapping}
\end{figure}
%
\subsection{Numerical Results}
We consider the same network introduced in Section \ref{subsec:DL_Numerical}. In addition, flexible duplex is enabled such that resources can be dynamically allocated to \ac{UL} or \ac{DL} services. 
Owing to the space limitation, we only present one exemplary numerical result in Fig. \ref{fig:traffic}.  We compare the performance of the following four protocols: (\Rmnum{1})\lq\lq{FIX}\rq\rq \ for fixed ratio between \ac{UL} and \ac{DL} resources; (\Rmnum{2})\lq\lq{SAFP}\rq\rq \ for dynamic \ac{UL}/\ac{DL} resource configuration with \ac{SAFP} without muting; (\Rmnum{3})\lq\lq{Resource muting based on $\ve{I}$}\rq\rq \ for the resource muting scheme based on interference indicator $\ve{I}$ as introduced in Section \ref{subsec:DL_Numerical}, and (\Rmnum{4})\lq\lq{Resource muting based on $\ve{w}_{\infty}$}\rq\rq \ for the resource muting scheme based on the asymptotic behavior $\vw_{\infty}$. The performance  is obtained by averaging $1000$ Monte Carlo simulations with random distribution of user locations and traffic demands conditioned on the given traffic asymmetry for $K = 200$. We define a measure called {\it inter-cell traffic distance} to reflect both the inter-cell traffic asymmetry and the intra-cell \ac{UL}/\ac{DL} traffic asymmetry between a pair of cells $m,n$, computed as $D_{m,n} := \|\ve{d}_n - \ve{d}_m\|_2$, where $\|\cdot\|_{2}$ denotes the $L^{2}$-norm, and the per-cell \ac{UL}/\ac{DL} traffic demands $\ve{d}_n:=\left[d_n^{\ul}, d_n^{\dl}\right]$ are normalized. Fig. \ref{fig:traffic} shows that dynamic \ac{UL}/\ac{DL} configuration achieves a twofold increase in the average utility compared to fixed \ac{UL}/\ac{DL} ratio. Resource muting brings further improvement, varying from $20$\% to  $100$\%, by adapting to the traffic asymmetry. Note that the performance gains increase with the traffic asymmetry. 
\begin{figure}[t]
    \centering
        \includegraphics[width=1\columnwidth]{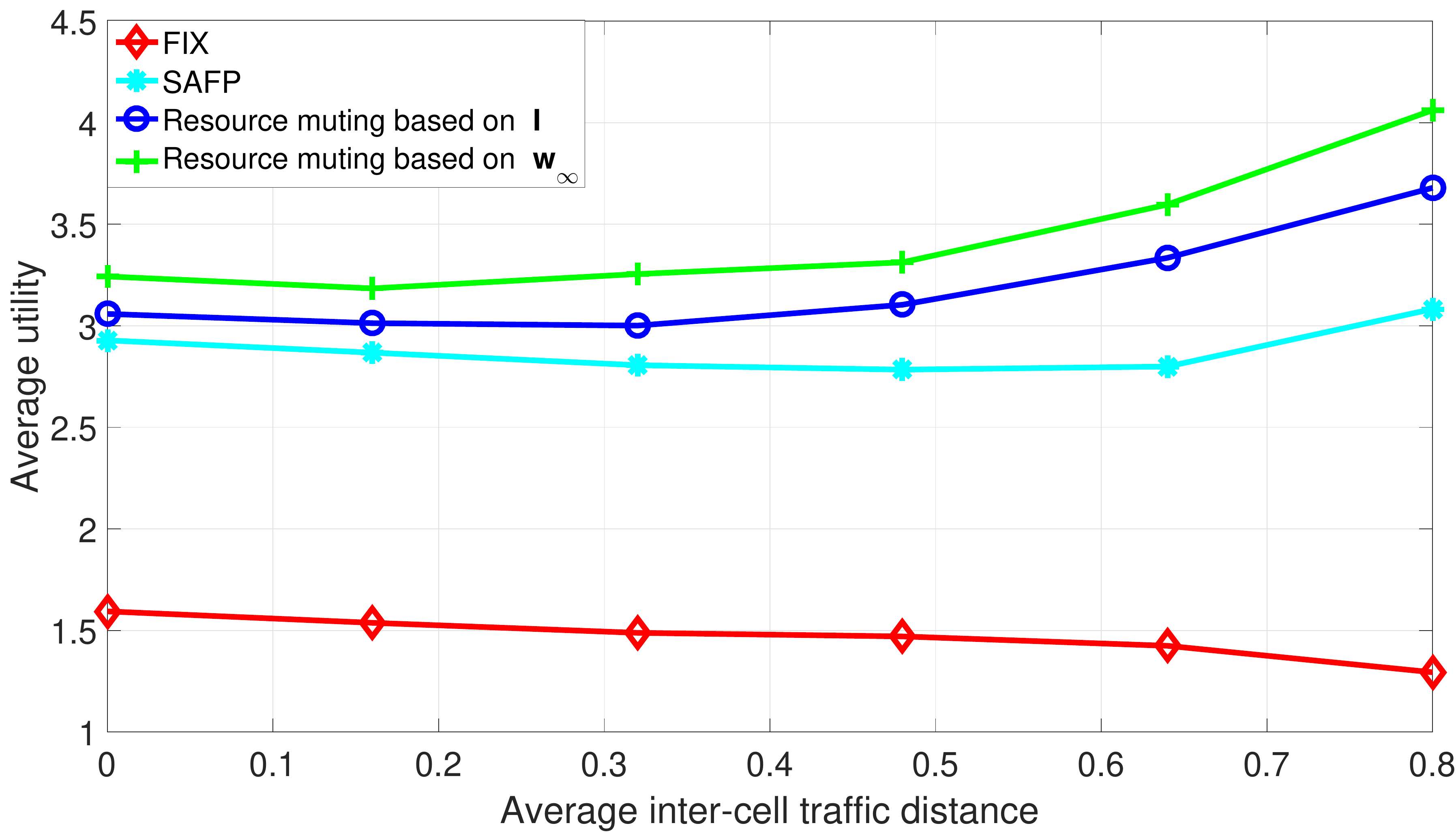}
        \caption{Average utility with different inter-cell traffic distances.}
 \label{fig:traffic}
\vspace{-2ex}
\end{figure}
\section{Conclusion}\label{sec:concl}
We have characterized an efficient resource utilization region by studying the asymptotic behavior of solutions to max-min utility optimization problems with  interference models based on the load imposed by services. Building upon this result, we have developed a partial resource muting scheme that is suitable for both conventional 4G \ac{DL} networks and flexible duplex enabled 5G networks. Simulations show significant performance gains in both scenarios, measured by the worst-case \ac{QoS} satisfaction level, compared to the optimal solution to the max-min optimization problem without resource muting mechanisms. We have also shown that the gain obtained with the proposed muting scheme increases when the traffic  becomes highly asymmetric. 

%
\acrodef{3GPP}{3rd generation partnership project}
\acrodef{5G}{fifth generation}
 \acrodef{ABS}{almost blank subframe}
    \acrodef{BS}{base station}
    \acrodef{CDF}{cumulative distribution function}
		\acrodef{CEVP}{constrained eigenvalue problem}
    \acrodef{CSI}{channel state information}
    \acrodef{CQI}{channel quality indicator}
\acrodef{DL}{downlink}
    \acrodef{DUDe}{downlink and uplink decoupling}
\acrodef{eICIC}{enhanced intercell interference coordination}
\acrodef{ESD}{energy spectral density}
\acrodef{FDD}{frequency division duplex}
    \acrodef{FDMA}{frequency division multiple access}
   \acrodef{GP}{Gaussian process}
    \acrodef{GPS}{global positioning system}
\acrodef{HetNet}{heterogeneous network}
    \acrodef{ICI}{inter-cell interference}
		\acrodef{IMI}{inter-mode interference}
\acrodef{LTE}{long term evolution}

\acrodef{MAC}{media access control}
\acrodef{MIMO}{multiple-input and multiple-output}
\acrodef{MRU}{minimum resource unit}
\acrodef{NLES}{nonlinear equation system}
   \acrodef{OFDM}{orthogonal frequency division multiplexing}
    \acrodef{PDF}{probability density function}
    \acrodef{PHY}{physical layer}
		\acrodef{PSD}{power spectral density}
    \acrodef{PRB}{physical resource block}
   \acrodef{QoE}{quality of experience}
    \acrodef{QoS}{quality of service}
    \acrodef{RAN}{radio access network}
		\acrodef{RB}{resource block}
		\acrodef{RBS}{removal of bottleneck services}
		\acrodef{RMDI}{resource muting for dominant interferer}
    \acrodef{RRM}{radio resource management}
		\acrodef{RU}{resource unit}
		\acrodef{RX}{receiver}
 \acrodef{SAFP}{successive approximation of fixed point}
    \acrodef{SDN}{software defined network}
    \acrodef{SNR}{signal-to-noise ratio}
    \acrodef{SINR}{signal-to-interference-plus-noise ratio}
\acrodef{SIR}{signal-to-interference ratio}
\acrodef{SIF}{standard interference function}
    \acrodef{SVM}{support vector machine}
    \acrodef{TCP}{transmission control protocol}
		\acrodef{TDD}{time division duplex}
    \acrodef{TDMA}{time division multiple access}
		\acrodef{TTI}{transmission time interval}
		\acrodef{TX}{transmitter}
		\acrodef{UE}{user equipment}
		\acrodef{UL}{uplink}
    \acrodef{WLAN}{wireless local area network}
%
\vspace{-1ex}
\bibliographystyle{IEEEtran}
\bibliography{refs}
		
\end{document}